\documentclass[twocolumn,10pt]{article}

\usepackage[top=2cm, bottom=2cm, left=2cm, right=2cm,a4paper]{geometry}

\usepackage{amsmath,amsthm}
\usepackage{bm}  
\usepackage{bbm}
\usepackage{verbatim}
\usepackage{xcolor}

\usepackage[backend=bibtex,sorting=none,style=phys,biblabel=brackets,backref=false,bibencoding=utf8,maxnames=20,eprint=true]{biblatex}
\makeatletter
\pretocmd{\blx@head@bibintoc}{\phantomsection}{}{\ddt}
\makeatother
\DefineBibliographyStrings{english}{%
  backrefpage = {page},
  backrefpages = {pages},
}
\makeatletter
\def\blx@bblfile@bibtex{%
  \blx@secinit
  \begingroup
  \blx@bblstart
\begingroup
\makeatletter
\@ifundefined{ver@biblatex.sty}
  {\@latex@error
     {Missing 'biblatex' package}
     {The bibliography requires the 'biblatex' package.}
      \aftergroup }
  {}
\endgroup

\entry{coles_entropic_2015}{article}{}
  \name{author}{4}{}{%
    {{}%
     {Coles}{C.}%
     {Patrick~J.}{P.~J.}%
     {}{}%
     {}{}}%
    {{}%
     {Berta}{B.}%
     {Mario}{M.}%
     {}{}%
     {}{}}%
    {{}%
     {Tomamichel}{T.}%
     {Marco}{M.}%
     {}{}%
     {}{}}%
    {{}%
     {Wehner}{W.}%
     {Stephanie}{S.}%
     {}{}%
     {}{}}%
  }
  \keyw{Quantum Physics}
  \strng{namehash}{CPJBMTMWS1}
  \strng{fullhash}{CPJBMTMWS1}
  \field{abstract}{%
  Heisenberg's uncertainty principle forms a fundamental element of quantum
  mechanics. Uncertainty relations in terms of entropies were initially
  proposed to deal with conceptual shortcomings in the original formulation of
  the uncertainty principle, and hence play an important role in quantum
  foundations. More recently, entropic uncertainty relations have emerged as
  the central ingredient in the security analysis of almost all quantum
  cryptographic protocols, ranging from quantum key distribution to two-party
  quantum cryptography. This review surveys entropic uncertainty relations that
  capture Heisenberg's idea that the results of incompatible measurements are
  impossible to predict, covering both finite- and infinite-dimensional
  measurements. These ideas are then extended to incorporate quantum
  correlations between the observed object and its environment, allowing for a
  variety of recent, more general formulations of the uncertainty principle.
  Finally various applications are discussed, ranging from entanglement
  witnessing to wave-particle duality to quantum cryptography.%
  }
  \verb{eprint}
  \verb 1511.04857
  \endverb
  \field{title}{Entropic {{Uncertainty Relations}} and their {{Applications}}}
  \verb{url}
  \verb http://arxiv.org/abs/1511.04857
  \endverb
  \verb{file}
  \verb arXiv\:1511.04857 PDF:/Users/joerenes/store/zotero/storage/UB8W8ZJS/Col
  \verb es et al. - 2015 - Entropic Uncertainty Relations and their Applicati.p
  \verb df:application/pdf
  \endverb
  \field{eprinttype}{arxiv}
  \field{eprintclass}{quant-ph}
  \field{day}{16}
  \field{month}{11}
  \field{year}{2015}
  \field{urlday}{15}
  \field{urlmonth}{03}
  \field{urlyear}{2016}
  \warn{\item Can't use 'eprinttype' + 'archiveprefix'}
\endentry

\entry{berta_uncertainty_2010}{article}{}
  \name{author}{5}{}{%
    {{}%
     {Berta}{B.}%
     {Mario}{M.}%
     {}{}%
     {}{}}%
    {{}%
     {Christandl}{C.}%
     {Matthias}{M.}%
     {}{}%
     {}{}}%
    {{}%
     {Colbeck}{C.}%
     {Roger}{R.}%
     {}{}%
     {}{}}%
    {{}%
     {Renes}{R.}%
     {Joseph~M.}{J.~M.}%
     {}{}%
     {}{}}%
    {{}%
     {Renner}{R.}%
     {Renato}{R.}%
     {}{}%
     {}{}}%
  }
  \strng{namehash}{BMCMCRRJMRR1}
  \strng{fullhash}{BMCMCRRJMRR1}
  \field{abstract}{%
  The uncertainty principle, originally formulated by Heisenberg, clearly
  illustrates the difference between classical and quantum mechanics. The
  principle bounds the uncertainties about the outcomes of two incompatible
  measurements, such as position and momentum, on a particle. It implies that
  one cannot predict the outcomes for both possible choices of measurement to
  arbitrary precision, even if information about the preparation of the
  particle is available in a classical memory. However, if the particle is
  prepared entangled with a quantum memory, a device that might be available in
  the not-too-distant future, it is possible to predict the outcomes for both
  measurement choices precisely. Here, we extend the uncertainty principle to
  incorporate this case, providing a lower bound on the uncertainties, which
  depends on the amount of entanglement between the particle and the quantum
  memory. We detail the application of our result to witnessing entanglement
  and to quantum key distribution.%
  }
  \verb{doi}
  \verb 10.1038/nphys1734
  \endverb
  \verb{eprint}
  \verb 0909.0950
  \endverb
  \field{issn}{1745-2481}
  \field{pages}{659}
  \field{shortjournal}{Nat Phys}
  \field{title}{The uncertainty principle in the presence of quantum memory}
  \verb{url}
  \verb http://dx.doi.org/10.1038/nphys1734
  \endverb
  \field{volume}{6}
  \verb{file}
  \verb Berta et al. - 2010 - The uncertainty principle in the presence of quan
  \verb t.pdf:/Users/joerenes/store/zotero/storage/4CVSR2BF/Berta et al. - 2010
  \verb  - The uncertainty principle in the presence of quant.pdf:application/p
  \verb df;nphys1734-s1.pdf:/Users/joerenes/store/zotero/storage/FJA4HTXD/nphys
  \verb 1734-s1.pdf:application/pdf
  \endverb
  \field{journaltitle}{Nature Physics}
  \field{eprinttype}{arxiv}
  \field{eprintclass}{quant-ph}
  \field{day}{25}
  \field{month}{07}
  \field{year}{2010}
  \field{urlday}{09}
  \field{urlmonth}{09}
  \field{urlyear}{2010}
  \warn{\item Can't use 'eprinttype' + 'archiveprefix'}
\endentry

\entry{christandl_uncertainty_2005}{article}{}
  \name{author}{2}{}{%
    {{}%
     {Christandl}{C.}%
     {M.}{M.}%
     {}{}%
     {}{}}%
    {{}%
     {Winter}{W.}%
     {A.}{A.}%
     {}{}%
     {}{}}%
  }
  \strng{namehash}{CMWA1}
  \strng{fullhash}{CMWA1}
  \field{abstract}{%
  Squashed entanglement and entanglement of purification are quantum-mechanical
  correlation measures and are defined as certain minimizations of entropic
  quantities. In this paper, we present the first nontrivial calculations of
  both quantities. Our results lead to the conclusion that both measures can
  drop by an arbitrary amount when only a single qubit of a local system is
  lost. This property is known as “locking” and has previously been
  observed for other correlation measures such as accessible information,
  entanglement cost, and logarithmic negativity. In the case of squashed
  entanglement, the results are obtained using an inequality that can be
  understood as a quantum channel analogue of well-known entropic uncertainty
  relations. This inequality may prove a useful tool in quantum information
  theory. The regularized entanglement of purification is known to equal the
  entanglement needed to prepare many copies of a quantum state by local
  operations and a sublinear amount of communication. Here, monogamy of quantum
  entanglement (i.e., the impossibility of a system being maximally entangled
  with two others at the same time) leads to an exact calculation for all
  quantum states that are supported either on the symmetric or on the
  antisymmetric subspace of a\textdollar{}dtimes d\textdollar{}-dimensional
  system.%
  }
  \verb{doi}
  \verb 10.1109/TIT.2005.853338
  \endverb
  \verb{eprint}
  \verb quant-ph/0501090
  \endverb
  \field{issn}{0018-9448}
  \field{number}{9}
  \field{pages}{3159}
  \field{shortjournal}{Information Theory, IEEE Transactions on}
  \field{title}{Uncertainty, {{Monogamy}}, and {{Locking}} of {{Quantum
  Correlations}}}
  \field{volume}{51}
  \verb{file}
  \verb Christandl and Winter - 2005 - Uncertainty, Monogamy, and Locking of Qu
  \verb antum Corr.pdf:/Users/joerenes/store/zotero/storage/N28MDJZB/Christandl
  \verb  and Winter - 2005 - Uncertainty, Monogamy, and Locking of Quantum Corr
  \verb .pdf:application/pdf
  \endverb
  \field{journaltitle}{IEEE Transactions on Information Theory}
  \field{eprinttype}{arxiv}
  \field{year}{2005}
  \warn{\item Can't use 'eprinttype' + 'archiveprefix'}
\endentry

\entry{renes_conjectured_2009}{article}{}
  \name{author}{2}{}{%
    {{}%
     {Renes}{R.}%
     {Joseph~M.}{J.~M.}%
     {}{}%
     {}{}}%
    {{}%
     {Boileau}{B.}%
     {Jean-Christian}{J.-C.}%
     {}{}%
     {}{}}%
  }
  \strng{namehash}{RJMBJC1}
  \strng{fullhash}{RJMBJC1}
  \field{abstract}{%
  We conjecture a new entropic uncertainty principle governing the entropy of
  complementary observations made on a system given side information in the
  form of quantum states, generalizing the entropic uncertainty relation of
  Maassen and Uffink [Phys. Rev. Lett. 60, 1103 (1988)]. We prove a special
  case for certain conjugate observables by adapting a similar result found by
  Christandl and Winter pertaining to quantum channels [IEEE Trans. Inf. Theory
  51, 3159 (2005)], and discuss possible applications of this result to the
  decoupling of quantum systems and for security analysis in quantum
  cryptography.%
  }
  \verb{doi}
  \verb 10.1103/PhysRevLett.103.020402
  \endverb
  \verb{eprint}
  \verb 0806.3984
  \endverb
  \field{number}{2}
  \field{pages}{020402}
  \field{shortjournal}{Phys. Rev. Lett.}
  \field{title}{Conjectured {{Strong Complementary Information Tradeoff}}}
  \verb{url}
  \verb http://link.aps.org/abstract/PRL/v103/e020402
  \endverb
  \field{volume}{103}
  \verb{file}
  \verb Renes and Boileau - 2009 - Conjectured Strong Complementary Information
  \verb  Trade.pdf:/Users/joerenes/store/zotero/storage/46EMWGC8/Renes and Boil
  \verb eau - 2009 - Conjectured Strong Complementary Information Trade.pdf:app
  \verb lication/pdf
  \endverb
  \field{journaltitle}{Physical Review Letters}
  \field{eprinttype}{arxiv}
  \field{eprintclass}{quant-ph}
  \field{day}{10}
  \field{month}{07}
  \field{year}{2009}
  \field{urlday}{07}
  \field{urlmonth}{09}
  \field{urlyear}{2009}
  \warn{\item Can't use 'eprinttype' + 'archiveprefix'}
\endentry

\entry{tomamichel_uncertainty_2011}{article}{}
  \name{author}{2}{}{%
    {{}%
     {Tomamichel}{T.}%
     {Marco}{M.}%
     {}{}%
     {}{}}%
    {{}%
     {Renner}{R.}%
     {Renato}{R.}%
     {}{}%
     {}{}}%
  }
  \strng{namehash}{TMRR1}
  \strng{fullhash}{TMRR1}
  \field{abstract}{%
  Uncertainty relations give upper bounds on the accuracy by which the outcomes
  of two incompatible measurements can be predicted. While established
  uncertainty relations apply to cases where the predictions are based on
  purely classical data (e.g., a description of the system’s state before
  measurement), an extended relation which remains valid in the presence of
  quantum information has been proposed recently [Berta et al., Nature Phys. 6,
  659 (2010)]. Here, we generalize this uncertainty relation to one formulated
  in terms of smooth entropies. Since these entropies measure operational
  quantities such as extractable secret key length, our uncertainty relation is
  of immediate practical use. To illustrate this, we show that it directly
  implies security of quantum key distribution protocols. Our security claim
  remains valid even if the implemented measurement devices deviate arbitrarily
  from the theoretical model.%
  }
  \verb{doi}
  \verb 10.1103/PhysRevLett.106.110506
  \endverb
  \verb{eprint}
  \verb 1009.2015
  \endverb
  \field{number}{11}
  \field{pages}{110506}
  \field{shortjournal}{Phys. Rev. Lett.}
  \field{title}{Uncertainty {{Relation}} for {{Smooth Entropies}}}
  \verb{url}
  \verb http://link.aps.org/doi/10.1103/PhysRevLett.106.110506
  \endverb
  \field{volume}{106}
  \verb{file}
  \verb Tomamichel and Renner - 2011 - Uncertainty Relation for Smooth Entropie
  \verb s.pdf:/Users/joerenes/store/zotero/storage/TAKTQ5R3/Tomamichel and Renn
  \verb er - 2011 - Uncertainty Relation for Smooth Entropies.pdf:application/p
  \verb df
  \endverb
  \field{journaltitle}{Physical Review Letters}
  \field{eprinttype}{arxiv}
  \field{eprintclass}{quant-ph}
  \field{day}{16}
  \field{month}{03}
  \field{year}{2011}
  \field{urlday}{28}
  \field{urlmonth}{05}
  \field{urlyear}{2011}
  \warn{\item Can't use 'eprinttype' + 'archiveprefix'}
\endentry

\entry{coles_unifying_2015}{article}{}
  \name{author}{1}{}{%
    {{}%
     {Coles}{C.}%
     {Patrick~J.}{P.~J.}%
     {}{}%
     {}{}}%
  }
  \keyw{Quantum Physics}
  \strng{namehash}{CPJ1}
  \strng{fullhash}{CPJ1}
  \field{abstract}{%
  An interferometer - no matter how clever the design - cannot allow the
  experimenter to see both the wave and particle behavior of a quantum system.
  This fundamental idea has been captured quantitatively by inequalities,
  so-called wave-particle duality relations (WPDRs), that upper bound the sum
  of the path distinguishability (particle behavior) and fringe visibility
  (wave behavior). Another fundamental quantum concept is Heisenberg's
  uncertainty principle, stating that some pairs of observables cannot be known
  simultaneously. The distinction between these two concepts has been debated
  in the literature. Here we provide closure to the debate by showing that
  WPDRs correspond to a modern formulation of the uncertainty principle,
  namely, the uncertainty relation for the min- and max-entropies, which is
  used in quantum cryptography. At the technical level, our unification relies
  on a novel operational meaning for the max-entropy in terms of the guessing
  probability. Furthermore, our unification provides a framework for solving an
  outstanding problem of how to rigorously prove WPDRs for interferometers with
  more than two paths, and we employ this framework to derive some novel
  WPDRs.%
  }
  \verb{eprint}
  \verb 1512.09081
  \endverb
  \field{shorttitle}{Unifying wave-particle duality with entropic uncertainty}
  \field{title}{Unifying wave-particle duality with entropic uncertainty:
  {{Multi}}-path interferometers}
  \verb{url}
  \verb http://arxiv.org/abs/1512.09081
  \endverb
  \verb{file}
  \verb Coles_2015_Unifying wave-particle duality with entropic uncertainty.pdf
  \verb :/Users/joerenes/store/zotero/storage/6WF3AQNV/Coles_2015_Unifying wave
  \verb -particle duality with entropic uncertainty.pdf:application/pdf
  \endverb
  \field{eprinttype}{arxiv}
  \field{eprintclass}{quant-ph}
  \field{day}{30}
  \field{month}{12}
  \field{year}{2015}
  \field{urlday}{09}
  \field{urlmonth}{05}
  \field{urlyear}{2016}
  \warn{\item Can't use 'eprinttype' + 'archiveprefix'}
\endentry

\entry{renes_efficient_2012}{article}{}
  \name{author}{3}{}{%
    {{}%
     {Renes}{R.}%
     {Joseph~M.}{J.~M.}%
     {}{}%
     {}{}}%
    {{}%
     {Dupuis}{D.}%
     {Frédéric}{F.}%
     {}{}%
     {}{}}%
    {{}%
     {Renner}{R.}%
     {Renato}{R.}%
     {}{}%
     {}{}}%
  }
  \strng{namehash}{RJMDFRR1}
  \strng{fullhash}{RJMDFRR1}
  \field{abstract}{%
  Polar coding, introduced 2008 by Arıkan, is the first (very) efficiently
  encodable and decodable coding scheme whose information transmission rate
  provably achieves the Shannon bound for classical discrete memoryless
  channels in the asymptotic limit of large block sizes. Here, we study the use
  of polar codes for the transmission of quantum information. Focusing on the
  case of qubit Pauli channels and qubit erasure channels, we use classical
  polar codes to construct a coding scheme that asymptotically achieves a net
  transmission rate equal to the coherent information using efficient encoding
  and decoding operations and code construction. Our codes generally require
  preshared entanglement between sender and receiver, but for channels with a
  sufficiently low noise level we demonstrate that the rate of preshared
  entanglement required is zero.%
  }
  \verb{doi}
  \verb 10.1103/PhysRevLett.109.050504
  \endverb
  \verb{eprint}
  \verb 1109.3195
  \endverb
  \field{number}{5}
  \field{pages}{050504}
  \field{shortjournal}{Phys. Rev. Lett.}
  \field{title}{Efficient {{Polar Coding}} of {{Quantum Information}}}
  \verb{url}
  \verb http://link.aps.org/doi/10.1103/PhysRevLett.109.050504
  \endverb
  \field{volume}{109}
  \verb{file}
  \verb Renes et al. - 2012 - Efficient Polar Coding of Quantum Information.pdf
  \verb :/Users/joerenes/store/zotero/storage/BK4F2PXE/Renes et al. - 2012 - Ef
  \verb ficient Polar Coding of Quantum Information.pdf:application/pdf;polar-p
  \verb rl-supplement.pdf:/Users/joerenes/store/zotero/storage/QMHK67QB/polar-p
  \verb rl-supplement.pdf:application/pdf
  \endverb
  \field{journaltitle}{Physical Review Letters}
  \field{eprinttype}{arxiv}
  \field{eprintclass}{quant-ph}
  \field{day}{01}
  \field{month}{08}
  \field{year}{2012}
  \field{urlday}{16}
  \field{urlmonth}{08}
  \field{urlyear}{2012}
  \warn{\item Can't use 'eprinttype' + 'archiveprefix'}
\endentry

\entry{nielsen_quantum_2000}{book}{}
  \name{author}{2}{}{%
    {{}%
     {Nielsen}{N.}%
     {Michael~A.}{M.~A.}%
     {}{}%
     {}{}}%
    {{}%
     {Chuang}{C.}%
     {Isaac~L.}{I.~L.}%
     {}{}%
     {}{}}%
  }
  \list{publisher}{1}{%
    {{Cambridge University Press}}%
  }
  \strng{namehash}{NMACIL1}
  \strng{fullhash}{NMACIL1}
  \field{isbn}{0-521-63503-9}
  \field{title}{Quantum {{Computation}} and {{Quantum Information}}}
  \verb{file}
  \verb Nielsen and Chuang - 2000 - Quantum Computation and Quantum Information
  \verb .pdf:/Users/joerenes/store/zotero/storage/H3H688JB/Nielsen and Chuang -
  \verb  2000 - Quantum Computation and Quantum Information.pdf:application/pdf
  \endverb
  \field{year}{2000}
\endentry

\entry{leung_approximate_1997}{article}{}
  \name{author}{4}{}{%
    {{}%
     {Leung}{L.}%
     {Debbie~W.}{D.~W.}%
     {}{}%
     {}{}}%
    {{}%
     {Nielsen}{N.}%
     {M.~A.}{M.~A.}%
     {}{}%
     {}{}}%
    {{}%
     {Chuang}{C.}%
     {Isaac~L.}{I.~L.}%
     {}{}%
     {}{}}%
    {{}%
     {Yamamoto}{Y.}%
     {Yoshihisa}{Y.}%
     {}{}%
     {}{}}%
  }
  \strng{namehash}{LDWNMACILYY1}
  \strng{fullhash}{LDWNMACILYY1}
  \field{abstract}{%
  We present relaxed criteria for quantum error correction that are useful when
  the specific dominant quantum noise process is known. As an example, we
  provide a four-bit code that corrects for a single amplitude damping error.
  This code violates the usual Hamming bound calculated for a Pauli description
  of the error process and has no simple explanation in terms of the usual
  Pauli basis GF(4) codes.%
  }
  \verb{doi}
  \verb 10.1103/PhysRevA.56.2567
  \endverb
  \verb{eprint}
  \verb quant-ph/9704002
  \endverb
  \field{number}{4}
  \field{pages}{2567\bibrangedash 2573}
  \field{shortjournal}{Phys. Rev. A}
  \field{title}{Approximate quantum error correction can lead to better codes}
  \verb{url}
  \verb http://link.aps.org/doi/10.1103/PhysRevA.56.2567
  \endverb
  \field{volume}{56}
  \verb{file}
  \verb Leung et al. - 1997 - Approximate quantum error correction can lead to
  \verb b.pdf:/Users/joerenes/store/zotero/storage/IIG8F65W/Leung et al. - 1997
  \verb  - Approximate quantum error correction can lead to b.pdf:application/p
  \verb df
  \endverb
  \field{journaltitle}{Physical Review A}
  \field{eprinttype}{arxiv}
  \field{day}{01}
  \field{month}{10}
  \field{year}{1997}
  \field{urlday}{25}
  \field{urlmonth}{09}
  \field{urlyear}{2014}
  \warn{\item Can't use 'eprinttype' + 'archiveprefix'}
\endentry

\entry{alvaro}{article}{}
  \name{author}{2}{}{%
    {{}%
     {Piedrafita}{P.}%
     {\'{A}lvaro}{A.}%
     {}{}%
     {}{}}%
    {{}%
     {Renes}{R.}%
     {Joseph~M.}{J.~M.}%
     {}{}%
     {}{}}%
  }
  \strng{namehash}{PARJM1}
  \strng{fullhash}{PARJM1}
  \field{title}{Channel adapted decoding strategies based on complementarity}
  \field{journaltitle}{in preparation}
\endentry

\entry{devetak_private_2005}{article}{}
  \name{author}{1}{}{%
    {{}%
     {Devetak}{D.}%
     {Igor}{I.}%
     {}{}%
     {}{}}%
  }
  \keyw{channel capacity,channel coding,channel coding theorem,coherent
  information,cryptography,key generation,private classical information
  transmission,protocols,public classical communication,public key
  cryptography,pure bipartite entanglement,quantum channel capacity,quantum
  cryptography,quantum entanglement,wiretap channel}
  \strng{namehash}{DI1}
  \strng{fullhash}{DI1}
  \field{abstract}{%
  A formula for the capacity of a quantum channel for transmitting private
  classical information is derived. This is shown to be equal to the capacity
  of the channel for generating a secret key, and neither capacity is enhanced
  by forward public classical communication. Motivated by the work of
  Schumacher and Westmoreland on quantum privacy and quantum coherence,
  parallels between private classical information and quantum information are
  exploited to obtain an expression for the capacity of a quantum channel for
  generating pure bipartite entanglement. The latter implies a new proof of the
  quantum channel coding theorem and a simple proof of the converse. The
  coherent information plays a role in all of the above mentioned capacities.%
  }
  \verb{doi}
  \verb 10.1109/TIT.2004.839515
  \endverb
  \verb{eprint}
  \verb quant-ph/0304127
  \endverb
  \field{issn}{0018-9448}
  \field{number}{1}
  \field{pages}{44}
  \field{title}{The private classical capacity and quantum capacity of a
  quantum channel}
  \field{volume}{51}
  \verb{file}
  \verb Devetak - 2005 - The private classical capacity and quantum capacit.pdf
  \verb :/Users/joerenes/store/zotero/storage/MM8TN2HU/Devetak - 2005 - The pri
  \verb vate classical capacity and quantum capacit.pdf:application/pdf
  \endverb
  \field{journaltitle}{IEEE Transactions on Information Theory}
  \field{eprinttype}{arxiv}
  \field{year}{2005}
  \warn{\item Can't use 'eprinttype' + 'archiveprefix'}
\endentry

\entry{schumacher_approximate_2002}{article}{}
  \name{author}{2}{}{%
    {{}%
     {Schumacher}{S.}%
     {Benjamin}{B.}%
     {}{}%
     {}{}}%
    {{}%
     {Westmoreland}{W.}%
     {Michael~D.}{M.~D.}%
     {}{}%
     {}{}}%
  }
  \strng{namehash}{SBWMD1}
  \strng{fullhash}{SBWMD1}
  \field{abstract}{%
  The errors that arise in a quantum channel can be corrected perfectly if and
  only if the channel does not decrease the coherent information of the input
  state. We show that, if the loss of coherent information is small, then
  approximate quantum error correction is possible.%
  }
  \verb{doi}
  \verb 10.1023/A:1019653202562
  \endverb
  \verb{eprint}
  \verb quant-ph/0112106
  \endverb
  \field{number}{1}
  \field{pages}{5\bibrangedash 12}
  \field{title}{Approximate {{Quantum Error Correction}}}
  \verb{url}
  \verb http://dx.doi.org/10.1023/A:1019653202562
  \endverb
  \field{volume}{1}
  \verb{file}
  \verb Schumacher and Westmoreland - 2002 - Approximate Quantum Error Correcti
  \verb on.pdf:/Users/joerenes/store/zotero/storage/U5W6PPJ8/Schumacher and Wes
  \verb tmoreland - 2002 - Approximate Quantum Error Correction.pdf:application
  \verb /pdf
  \endverb
  \field{journaltitle}{Quantum Information Processing}
  \field{eprinttype}{arxiv}
  \field{day}{01}
  \field{month}{04}
  \field{year}{2002}
  \field{urlday}{24}
  \field{urlmonth}{09}
  \field{urlyear}{2007}
  \warn{\item Can't use 'eprinttype' + 'archiveprefix'}
\endentry

\entry{tomamichel_monogamyofentanglement_2013}{article}{}
  \name{author}{4}{}{%
    {{}%
     {Tomamichel}{T.}%
     {Marco}{M.}%
     {}{}%
     {}{}}%
    {{}%
     {Fehr}{F.}%
     {Serge}{S.}%
     {}{}%
     {}{}}%
    {{}%
     {Kaniewski}{K.}%
     {Jędrzej}{J.}%
     {}{}%
     {}{}}%
    {{}%
     {Wehner}{W.}%
     {Stephanie}{S.}%
     {}{}%
     {}{}}%
  }
  \strng{namehash}{TMFSKJWS1}
  \strng{fullhash}{TMFSKJWS1}
  \field{abstract}{%
  We consider a game in which two separate laboratories collaborate to prepare
  a quantum system and are then asked to guess the outcome of a measurement
  performed by a third party in a random basis on that system. Intuitively, by
  the uncertainty principle and the monogamy of entanglement, the probability
  that both players simultaneously succeed in guessing the outcome correctly is
  bounded. We are interested in the question of how the success probability
  scales when many such games are performed in parallel. We show that any
  strategy that maximizes the probability to win every game individually is
  also optimal for the parallel repetition of the game. Our result implies that
  the optimal guessing probability can be achieved without the use of
  entanglement. We explore several applications of this result. Firstly, we
  show that it implies security for standard BB84 quantum key distribution when
  the receiving party uses fully untrusted measurement devices , i.e. we show
  that BB84 is one-sided device independent. Secondly, we show how our result
  can be used to prove security of a one-round position-verification scheme.
  Finally, we generalize a well-known uncertainty relation for the guessing
  probability to quantum side information.%
  }
  \verb{doi}
  \verb 10.1088/1367-2630/15/10/103002
  \endverb
  \verb{eprint}
  \verb 1210.4359
  \endverb
  \field{issn}{1367-2630}
  \field{number}{10}
  \field{pages}{103002}
  \field{shortjournal}{New J. Phys.}
  \field{title}{A monogamy-of-entanglement game with applications to
  device-independent quantum cryptography}
  \verb{url}
  \verb http://stacks.iop.org/1367-2630/15/i=10/a=103002
  \endverb
  \field{volume}{15}
  \field{langid}{english}
  \verb{file}
  \verb Tomamichel et al_2013_A monogamy-of-entanglement game with applications
  \verb  to device-independent quantum.pdf:/Users/joerenes/store/zotero/storage
  \verb /H4QTQCCH/Tomamichel et al_2013_A monogamy-of-entanglement game with ap
  \verb plications to device-independent quantum.pdf:application/pdf
  \endverb
  \field{journaltitle}{New Journal of Physics}
  \field{eprinttype}{arxiv}
  \field{eprintclass}{quant-ph}
  \field{year}{2013}
  \field{urlday}{09}
  \field{urlmonth}{11}
  \field{urlyear}{2015}
  \warn{\item Can't use 'eprinttype' + 'archiveprefix'}
\endentry

\entry{renes_heisenberg_2016}{article}{}
  \name{author}{3}{}{%
    {{}%
     {Renes}{R.}%
     {Joseph~M.}{J.~M.}%
     {}{}%
     {}{}}%
    {{}%
     {Scholz}{S.}%
     {Volkher~B.}{V.~B.}%
     {}{}%
     {}{}}%
    {{}%
     {Huber}{H.}%
     {Stefan}{S.}%
     {}{}%
     {}{}}%
  }
  \strng{namehash}{RJMSVBHS1}
  \strng{fullhash}{RJMSVBHS1}
  \field{title}{Heisenberg un- certainty relations: An operational approach}
  \field{journaltitle}{in preparation}
\endentry

\entry{kretschmann_information-disturbance_2008}{article}{}
  \name{author}{3}{}{%
    {{}%
     {Kretschmann}{K.}%
     {D.}{D.}%
     {}{}%
     {}{}}%
    {{}%
     {Schlingemann}{S.}%
     {D.}{D.}%
     {}{}%
     {}{}}%
    {{}%
     {Werner}{W.}%
     {R.F.}{R.}%
     {}{}%
     {}{}}%
  }
  \strng{namehash}{KDSDWR1}
  \strng{fullhash}{KDSDWR1}
  \field{abstract}{%
  Stinespring's dilation theorem is the basic structure theorem for quantum
  channels: it states that any quantum channel arises from a unitary evolution
  on a larger system. Here we prove a continuity theorem for Stinespring's
  dilation: if two quantum channels are close in cb-norm, then it is always
  possible to find unitary implementations which are close in operator norm,
  with dimension-independent bounds. This result generalizes Uhlmann's theorem
  from states to channels and allows to derive a formulation of the
  information-disturbance tradeoff in terms of quantum channels, as well as a
  continuity estimate for the no-broadcasting theorem. We briefly discuss
  further implications for quantum cryptography, thermalization processes, and
  the black hole information loss puzzle.%
  }
  \verb{doi}
  \verb 10.1109/TIT.2008.917696
  \endverb
  \verb{eprint}
  \verb quant-ph/0605009
  \endverb
  \field{issn}{0018-9448}
  \field{number}{4}
  \field{pages}{1708}
  \field{title}{The {{Information-Disturbance Tradeoff}} and the {{Continuity}}
  of {{Stinespring}}'s {{Representation}}}
  \field{volume}{54}
  \verb{file}
  \verb Kretschmann et al. - 2008 - The Information-Disturbance Tradeoff and th
  \verb e Conti.pdf:/Users/joerenes/store/zotero/storage/RU588MMT/Kretschmann e
  \verb t al. - 2008 - The Information-Disturbance Tradeoff and the Conti.pdf:a
  \verb pplication/pdf
  \endverb
  \field{journaltitle}{IEEE Transactions on Information Theory}
  \field{eprinttype}{arxiv}
  \field{year}{2008}
  \warn{\item Can't use 'eprinttype' + 'archiveprefix'}
\endentry

\lossort
\endlossort

  \blx@bblend
  \endgroup
  \csnumgdef{blx@labelnumber@\the\c@refsection}{0}}
\makeatother

\usepackage{titlesec}

\titleformat*{\section}{\bfseries}
\titleformat*{\subsection}{\normalsize\bfseries}
\titleformat*{\subsubsection}{\bfseries}
\titleformat*{\paragraph}{\large\bfseries}
\titleformat*{\subparagraph}{\large\bfseries}

\titlespacing\section{0pt}{12pt plus 4pt minus 2pt}{2pt plus 2pt minus 2pt}

\usepackage[bookmarks,colorlinks,breaklinks]{hyperref}  
\definecolor{dullmagenta}{rgb}{0.4,0,0.4}   
\definecolor{darkblue}{rgb}{0,0,0.4}
\hypersetup{linkcolor=red,citecolor=blue,filecolor=dullmagenta,urlcolor=darkblue} 

\newcommand\invisiblesection[1]{%
  \addcontentsline{toc}{section}{#1}%
  \sectionmark{#1}
  }

\RequirePackage[charter, greekuppercase=italicized]{mathdesign}
\RequirePackage{beramono}
\RequirePackage{berasans}

\usepackage{amsbsy,verbatim,graphicx}

\newcommand{\ket}[1]{\left|#1\right\rangle}

\newcommand{\bra}[1]{\left\langle #1\right|}

\newcommand{\ketbra}[1]{\ket{#1}\bra{#1}}

\newcommand{\pg}{P}
\newcommand{\f}{F}
\newcommand{\q}{Q}
\newcommand{\tr}{{\text{Tr}}}
\newcommand{\id}{\mathbbm 1}
\newcommand{\mix}{\pi}
\newcommand{\Hk}{H_K}
\newcommand{\hHk}{\hat H_K}
\newcommand{\Dk}{D_K}
\newcommand{\acos}{\mathrm{acos}\,}

\newtheorem{theorem}{Theorem}
\newtheorem{corollary}{Corollary}

\newtheorem{lemma}{Lemma}

\usepackage{xfrac}

\usepackage{setspace}
\usepackage{titling}
\newcommand \myabstract[2][.8]{%
  \renewcommand\maketitlehookd{%
    \centering
    \begin{minipage}{#1\textwidth}
      {\begin{spacing}{1.0} \small #2\end{spacing}}
    \end{minipage}}}  

\begin{document}
\author{{\Large Joseph M.~Renes}\\
{ Institute for Theoretical Physics, ETH Z\"urich, 8093 Z\"urich, Switzerland}
}


\title{Uncertainty relations and approximate quantum error correction}

\date{}
\vspace{-1cm}

\myabstract{
The uncertainty principle can be understood as constraining the probability of winning a game in which Alice measures one of two conjugate observables, such as position or momentum, on a system provided by Bob, and he is to guess the outcome. 
Two variants are possible: either Alice tells Bob which observable she measured, or he has to furnish guesses for both cases. 
Here I derive new uncertainty relations for both, formulated directly in terms of Bob's guessing probabilities.
For the former these relate to the entanglement that can be recovered by action on Bob's system alone. 
This gives a condition for approximate quantum error correction in terms of the recoverability of ``amplitude'' and ``phase'' information, implicitly used in the recent construction of efficient  quantum polar codes. 
I also find a new relation on the guessing probabilities for the latter game, which 
has application to wave-particle duality relations.
}

\maketitle

\invisiblesection{Introduction}
Beyond their foundational appeal, uncertainty relations have become an important tool in quantum information theory, particularly entropic formulations (for a review, see~\cite{coles_entropic_2015}). 
One way to frame the recent statements is in terms of a guessing game~\cite{berta_uncertainty_2010}. 
In the game Bob prepares a quantum system and gives to Alice, who then measures one of two conjugate observables, such as position or momentum. 
She then asks Bob to guess the outcome of her measurement, and he wins the game if he can guess correctly. 
There are two versions of the game, depending on whether Alice tells Bob which observable she measured. 
If she does, then Bob need only furnish a guess for one observable, but if not he has no choice but to guess what the outcome would be in either case. 

According to the uncertainty principle, the latter variant must be impossible to win with any reliability, but the former game can be won if Bob supplies Alice with one half of a maximally entangled state. For instance, if the observables are orthogonal angular momentum components of a spin-1/2 particle, then Bob can win the game by supplying a spin-singlet state. Upon learning which observable Alice measured, he performs the same measurement on his spin and reports the opposite result. 
Quantitative bounds on the game formulated using entropy can be found in \cite{christandl_uncertainty_2005,renes_conjectured_2009} for conjugate observables, and subsequent work has generalized the statements to arbitrary observables and entropy measures.   

In this article I derive uncertainty relations for both variants of the game, formulated directly in terms of the guessing probabilities, rather than entropies. 
Guessing probability is a more directly operational quantity than entropy, which  yields more straightforward quantitative constraints on the guessing game and an even simpler interpretation of the resulting uncertainty relations, 
For the latter game, call it version 2, I find a bound which constrains the allowed combinations of guessing probabilities. 
In accordance with intuition from the uncertainty principle, it implies that if Bob can reliably guess one of the observables, then he can do little better than to randomly guess the other. 
The bound builds on a closely related uncertainty relation for min- and max-entropies~\cite{tomamichel_uncertainty_2011}, and turns out to be related to wave-particle duality relations~\cite{coles_unifying_2015}.

For the former game, version 1, the new uncertainty relation provides a converse to the sufficiency of using entanglement to win the game. It implies that if Bob can reliably guess either observable, then it must necessarily be possible to recover a high-fidelity entangled state by acting on his system alone. 
This can be viewed as a sufficient condition for approximate quantum error correction, and the decomposition into guessing two observables is reminiscent of the focus on ``amplitude'' and ``phase'' errors in constructions of exact quantum error correcting codes.
Here, however, the focus is shifted away from \emph{errors} and onto classical amplitude and phase \emph{information}, i.e.\ the information about the two observables. 
Thus, whenever a channel can reliably transmit both kinds of classical information, it can reliably transmit entanglement.
The appeal of this condition is that it reduces the quantum task to two simpler, more classical tasks, but does not require the picture of amplitude and phase errors. 

A similar observation was made in~\cite{christandl_uncertainty_2005}, quantifying reliability in terms of entropy.  
Here the link is more direct, however, as the proof proceeds by building an entanglement recovery map from Bob's guessing strategies, which can be viewed as  measurements.
This ensures that properties of the classical decoding measurements can be transferred to the entanglement recovery map. 
Indeed, this approximate error-correcting condition was implicitly used in the construction of quantum polar codes by the author and others~\cite{renes_efficient_2012}, and the constructive link is crucial in showing that the quantum codes are efficiently decodable for suitable channels.


\section{Uncertainty guessing games}
\noindent Let us restrict our attention to finite-dimensional quantum systems, fix a dimension $d$, and suppose that Alice measures one of the generalized Pauli operators $X=\sum_{z=0}^{d-1}\ket{z+1}\bra{z}$ or $Z=\sum_{z=0}^{d-1}\omega^z \ketbra z$. 
Here $\{\ket{z}\}$ is a fixed basis and $\omega$ is a primitive $d$th root of unity. 
Denote the eigenvectors of $X$ by $\ket{\tilde x}=\tfrac1{\sqrt d}\sum_{z=0}^{d-1}\omega^{xz}\ket{z}$. 
These observables are conjugate in the sense that any eigenstate of one has uniform overlap with any eigenstate of the other, namely $1/\sqrt{d}$. 

Bob is free to prepare any conceivable quantum state of many systems, call it $\psi$, and give a $d$-dimensional subsystem denoted $A$ to Alice for measurement. 
We can describe his procedure for guessing Alice's outcome as performing a measurement on his remaining systems, as was done in the spin-singlet example above. 
Generally, both guessing measurements are POVMs, and let us denote by $\Lambda_z$ ($\Gamma_x$) his POVM for guessing the $Z$ ($X$) outcome. 


The distinction between the two versions of the game is whether the two POVMs must be performed simultaneously or not, i.e.\ if they commute.  
They need not in version 1, when Alice tells Bob which observable she measured.
But in version 2 Alice demands both guesses, so Bob must perform both measurements. 
In the latter case it is convenient to regard the commutation of the POVMs as arising from the fact that they are measurements on different subsystems, call them $B$ and $E$. 

In either case we are chiefly interested in the optimal probability that Bob guesses correctly, which for the $Z$ observable is given by 
\begin{align}
\pg(Z^A|B)_\psi:=\max_{\Lambda_z} \tr\Big[\sum_z (\ketbra z^A\otimes \Lambda_z^B)\psi^{AB}\Big].
\end{align}
Here the optimization is over all valid POVMs, i.e.\ $d$ positive semidefinite operators $\Lambda_z$ with the property that $\sum_{z=0}^{d-1}\Lambda_z =\id$.
The optimal guessing probability for the $X$ observable is entirely analogous.
Also important is the maximum entanglement fidelity that can be obtained from $\psi$ by acting on Bob's systems alone,
\begin{align}
\f(A|B)_{\psi}:=\max_{\mathcal E^{A'|B}} F\left({\Phi}^{AA'},\mathcal E^{A'|B}(\psi^{AB})\right),
\end{align}
where $F(\rho,\sigma)=\|\sqrt\rho\sqrt\sigma\|_1$ is the fidelity. 
Here the maximum is over quantum channels $\mathcal E^{A'|B}$ taking $B$ to $A'\simeq A$ and $\Phi^{AA'}$ is any maximally-entangled state.

\section{Version 1: Noncommuting guesses}
\noindent Bob can win the bipartite game for any dimension $d$ by preparing an entangled state $\ket{\Phi}^{AB}=\tfrac1{\sqrt d}\sum_{z=0}^{d-1}\ket{z}^A\ket{z}^B$. 
No matter which observable Alice measures on $A$, Bob performs the same measurement on $B$. 
Clearly, for $Z$, Alice and Bob's measurement outcomes always match, and thus $P(Z^A|B)_\Phi=1$. 
The same conclusion holds for $P(X^A|B)_\Phi$, since a simple calculation shows that in this case their outcomes always sum to zero modulo $d$. 

In fact, entanglement in $\psi$ is necessary to win the game. 
This conclusion also holds approximately and is quantatively captured by the relation
\begin{align}
\label{eq:bipartiteresult1}
\acos \f(A|B)_\psi\leq \acos \pg(Z^A|B)_\psi+\acos \pg(X^A|B)_\psi,
\end{align}
where $\mathrm{acos}$ is the inverse of the cosine function. 



Let us now turn to the proof of \eqref{eq:bipartiteresult1}. 
We can actually show a more general statement, somewhat outside the scope of the game, but useful in the context of quantum error correction. 
It turns out that it is not strictly necessary for both guessing probabilites to be close to unity to conclude that entanglement can be recovered from $\psi$. 
We only need to show that $Z^A$ is recoverable from Bob's system and, separately, that $X^A$ is recoverable from Bob's system under the additional assumption that $Z$ is perfectly recoverable.  
More concretely, let $\psi_Z^{AA'B}=U_Z^{AA'|A}\psi^{AB}(U_Z^{AA'|A})^\dagger$, where $U_Z^{AA'|A}=\sum_{z=0}^{d-1}\ketbra z^A\otimes \ket z^{A'}$. 
Essentially, $U_Z$ copies the $Z$-value of $A$ to $A'$, and therefore $P(Z^A|A'B)_{\psi_Z}=1$. 
Then, we can show
\begin{theorem}
\label{bipartiteresult}
For any bipartite state $\psi^{AB}$, 
\begin{align}
\label{eq:bipartiteresult}
\mathrm{acos}\, \f(A|B)_\psi\leq \mathrm{acos}\, \pg(Z^A|B)_\psi+\mathrm{acos}\, \pg(X^A|BA')_{\psi_Z}.
\end{align}
\end{theorem}
Before proceeding to the proof, first note that \eqref{eq:bipartiteresult} implies \eqref{eq:bipartiteresult1} by monotonicity of $\mathrm{acos}$ and the following.
\begin{lemma}
\label{lem:connectpguess}
For any bipartite state $\psi^{AB}$, 
\begin{align}
\pg(X^A|BA')_{\psi_Z}\geq \pg(X^A|B)_\psi.
\end{align}
\end{lemma}
This statement is certainly plausible under the intuition that it is easier to guess $X$ when given $A'$ as well as $B$.
However, the process of generating $\psi_Z$ from $\psi$ alters the $X^A$ observable, and so this reasoning does not immediately apply.  
Nevertheless, the statement holds due to conjugacy of $X^A$ and $U_Z$. 
\begin{proof}
Start by noting that 
$U^{A'A|A}_Z
= \tfrac 1{\sqrt d}\sum_{x=0}^{d-1}\ket{\tilde x}^{A'}\otimes  (Z^{-x})^{A}.
$ 
Let $\ket{\psi}^{ABE}$ be a purification of $\psi^{AB}$ and $\ket{\psi_Z}^{AA'BE}=U_Z^{AA'|A}\ket{\psi}^{ABE}$. 
Since $A$ and $A'$ are interchangable in $\ket{\psi_Z}$, it follows that
\begin{align}
\label{eq:psizform}
\ket{\psi_Z}^{AA'BE}
&=\tfrac 1{\sqrt d}\sum_{x=0}^{d-1} \ket{\tilde x}^A\otimes (Z^{-x})^{A'}\ket{\psi}^{A'BE}.
\end{align}
The action of $Z^{-x}$ will be to shift the $X$ value of $A'$ by $-x$. 
But if the $X$ value of $A'$ in $\ket{\psi}^{A'BE}$ is recoverable from $B$, then by comparing the value on $A'$ and $B$, we can accurately determine the value of the shift. 
To this end, let $\Gamma_x^B$ be the optimal measurement in  $\pg(X^A|B)_\psi$. 
Define the new measurement with elements
$
\Xi_x^{A'B}=\sum_{x'=0}^{d-1}\widetilde \Pi_{x'-x}^{A'} \otimes \Gamma_{x'}^B$, 
where $\tilde \Pi_x$ is the projector onto $\ket{\tilde x}$. 
This measurement yields the difference between the outcome of the guessing measurement $\Gamma$ on $B$ and the $X$ measurement on $A'$. 
Notice that $(Z^{x})^{A'}\Xi_{x}^{A'B}(Z^{-x})^{A'}=\Xi_0^{A'B}$. 
Using the form of $\psi_Z$ in \eqref{eq:psizform} to compute $\pg(X^A|BA')_{\psi_Z}$  gives
\begin{subequations}
\begin{align}
\pg(X^A|BA')_{\psi_Z}
&\geq \tr\left[\left(\sum_{x=0}^{d-1}\widetilde \Pi_x^A\otimes \Xi_x^{A'B}\right)\psi_Z^{AA'BE}\right]\\
&=\bra{\psi}\Xi_0^{A'B}\ket{\psi}^{A'BE}\\
&=\pg(X^A|B)_\psi,
\end{align}
\end{subequations}
establishing the claim. 
%
\end{proof}

\begin{proof}[Proof of Theorem~\ref{bipartiteresult}]
First consider the properties of the coherent implementations of Alice's measurements. 
The coherent $Z$ measurement is given by $U_Z$, and the analogous $U_X$ is simply $U_X=\sum_{x=0}^{d-1}\ketbra{\tilde x}^A\otimes \ket x^{A''}$. 
Performing one after the other yields 
\begin{align}
U_X^{A''A|A}U_Z^{A'A|A}
&=\tfrac 1{\sqrt{d}}\sum_{x,z=0}^{d-1} \omega^{-xz} \ket{x}^{A''}\otimes \ket{z}^{A'}\otimes \ket{\tilde x}\bra{z}^A.
\end{align}
The phase $\omega^{-xz}$ can be removed by a controlled-phase operation $V^{A'A''}=\sum_{x=0}^{d-1}\ketbra x^{A''}\otimes (Z^x)^{A'}$. 
Defining $W^{A'A''|A}:=V^{A'A''}U_X^{A''A|A}U_Z^{A'A|A}$, we find 
\begin{align}
W^{A'A''|A}
&= \tfrac 1{\sqrt{d}}\sum_{x=0}^{d-1} \ket{x}^{A''}\otimes  \ket{\tilde x}^A \otimes \id^{A'|A}.
\end{align}
This operator transfers $A$ to $A'$ and then creates a maximally-entangled state in $AA''$. 

Thus, if Bob can simulate the action of $U_X$ and $U_Z$ by coherent measurements on his system, he should be able to create a high-fidelity entangled state in $AA''$. 
Suppose that the optimal measurements for guessing $Z$ and $X$ are $\Lambda_z^B$ and $\Gamma_x^{A'B}$, respectively. 
Define the coherent implementations of his two measurements,
$V_Z^{A'B|B}:=\sum_{z} \ket{z}^{A'}\otimes \sqrt{\Lambda_z^B}$ and $V_X^{A''A'B|A'B}:=\sum_x \ket{x}^{A''}\otimes \sqrt{\Gamma_x^{A'B}}$,
and consider the fidelity between $VV_XV_Z\ket{\psi}$ and $W\ket{\psi}$.
Since $VV_XV_Z$ is an operation solely on Bob's systems, we have
\begin{align}
\label{eq:recoverableent}
F(A|B)_\psi\geq F(W\ket{\psi},VV_XV_Z\ket{\psi}).
\end{align}
The operation of $VV_XV_Z$ is shown as a quantum circuit in Fig.~\ref{fig:entdec}. 
Using the triangle inequality and unitary invariance of the fidelity~\cite[\S9.2.2]{nielsen_quantum_2000}, we have
\begin{align}
\label{eq:fidelitytriangle}
&\acos F(W\ket{\psi},VV_XV_Z\ket{\psi})\\
&\,\,\leq \acos F(U_X\!\ket{\psi_Z},V_X\!\ket{\psi_Z}) + \acos F(\ket{\psi_Z},V_Z\!\ket{\psi}).\nonumber
\end{align}
We can bound the two terms, starting with the second. 
Notice that, since $\sqrt{\Lambda}\geq \Lambda$ for $0\leq \Lambda\leq \id$, 
\begin{align}
\label{eq:sqrtbound}
(U^{A'A|A}_Z)^\dagger V_Z^{A'B|B}
\geq \sum_{z=0}^{d-1} \ketbra z^A\otimes {\Lambda_z^B}.
\end{align}
Then $F(\ket{\psi_Z},V_Z\ket{\psi})\geq P(Z^A|B)_\psi$ immediately follows.
For $U_X^\dagger V_X$ the argument is entirely analogous and yields $F(U_X\ket{\psi_Z},V_X\ket{\psi_Z})\geq P(X^A|BA')_{\psi_Z}$. 
\end{proof}

\section{Approximate quantum error correction}
\noindent Both \eqref{eq:bipartiteresult1} and \eqref{eq:bipartiteresult} can be regarded as conditions for approximate quantum error correction. 
Suppose we are interested in transmitting entanglement through a given quantum channel by inputting one half of some fixed bipartite state. 
This results in an output state $\psi^{AB}$, and measuring either $X$ or $Z$ of system $A$ results in an output that corresponds to input of an $X$ or $Z$ eigenstate to the channel. 
With either set of $X$ or $Z$ inputs we could hope to send \emph{classical} information through the channel, and \eqref{eq:bipartiteresult1} or \eqref{eq:bipartiteresult} imply that if both of these classical tasks are reliable on average, then it is also possible to transmit \emph{quantum} information. 
Here the average is taken over the choice of $X$ or $Z$ inputs, the probabilities of which are determined by the associated measurement results.

\begin{figure}[th!]
\centering
\includegraphics{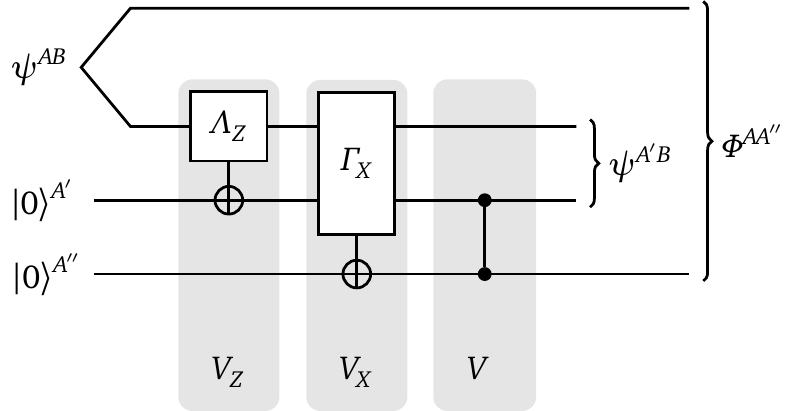}
\caption{\label{fig:entdec} The quantum circuit recovering entanglement from a bipartite state $\psi^{AB}$, when measurement of $Z$ or $X$ on $A$ can be reliably predicted by measurement of $B$ or $BA'$, respectively. 
The associated measurements $\Lambda_Z$ and $\Gamma_X$ are performed coherently in sequence, the latter taking results of the former into account, followed by a controlled-phase gate applied to the ancilla systems. 
The procedure also leaves the input state in systems $A'$ and $B$.
}

\end{figure}

As with the focus on amplitude and phase errors in exact quantum error correction~\cite[Ch.~10]{nielsen_quantum_2000}, \eqref{eq:bipartiteresult1} and \eqref{eq:bipartiteresult} allow us to break the problem of quantum transmission down into two classical pieces. 
This gives additional structure to the problem of designing encoding and decoding mechanisms and allows a large flexibility in adapting each to the particular channel at hand. 
This can help guide our search for reliable codes and encoders. 
And since we can lift efficient decoder constructions for classical transmission to quantum, this gives us some structure with which to construct efficient and practical decoders. 
As mentioned above, this link is crucial in constructing efficient quantum polar codes~\cite{renes_efficient_2012}.

Moreover, shifting the focus away from errors to information allows yet more flexibility in adapting an error-correction scheme to a particular channel. 
This can be illustrated in the original example of an approximate error-correcting code by Leung \emph{et al.}~\cite{leung_approximate_1997}, where just four physical qubits suffice to protect one encoded logical qubit from the action of the amplitude damping channel to first order in the damping probability. 
(Exact correction requires five qubits.)
Even though amplitude damping is not a combination of amplitude and phase errors, we can understand the operation of the approximate code as enabling reliable transmission of amplitude and phase information to first order.
Recently we have applied this approach to find structured decoders for approximate codes designed for the amplitude damping channel based on nonlinear classical codes~\cite{alvaro}.  



By altering the proof of Theorem~\ref{bipartiteresult}, we can obtain two stronger but nonconstructive conditions for approximate entanglement recovery. 
These both involve the fidelity
\begin{align}
\q(Z^A|B)_\psi:=F(\psi^{AB}_Z,\mix^Z\otimes \psi^B),
\end{align}
which quantifies how close the $Z$ outcome of $A$ is to being uniformly distributed and independent of the conditional state in $B$. 
\begin{theorem}
\label{thm:altbipartite}
For $\psi^{ABE}$ a pure state,
\begin{align}
\acos \f(A|B)_\psi&\leq \acos \pg(Z^A|B)_\psi+\acos \q(Z^A|E)_{\psi},\label{eq:aqecsecretkey}\\
\acos \f(A|B)_\psi&\leq \acos \q(X^A|A'E)_{\psi_Z}\!{+}\acos \q(Z^A|E)_{\psi}.\label{eq:aqectrump}
\end{align}
\end{theorem}
The first bound says that if the pure state $\psi^{ABE}$ can be used to create a secret key between Alice and Bob, a uniformly-distributed classical random variable independent of $E$, then the same state can be transformed into a maximally entangled state. A similar relation was used by Devetak in the achievability part of the quantum noisy channel coding theorem~\cite{devetak_private_2005}.
The second states that entanglement is recoverable from $B$ when system $E$ cannot predict $Z$ or $X$, the latter case even aided by knowledge of $Z$. 
This is broadly similar to the ``decoupling'' statement of Schumacher and Westmoreland~\cite{schumacher_approximate_2002}, but formulated as decoupling of $X$ and $Z$ information, not of the quantum state itself. 
The proof below  makes clear that $\q(Z^A|E)_\psi\geq \pg(X^A|A'B)_{\psi_Z}$ and likewise $\q(X^A|A'E)_{\psi_Z}\geq \pg(Z^A|B)_\psi$. 
Thus, the latter condition is the strongest as it implies the former, and both imply \eqref{eq:bipartiteresult}. 
\begin{proof}
The proof proceeds by replacing each of the two terms in the bound of \eqref{eq:fidelitytriangle} by fidelities. 

To establish \eqref{eq:aqecsecretkey} we use $\q(Z^A|E)_{\psi}$ to construct an appropriate $V_X$ in $F(U_X\ket{\psi_Z},V_X\ket{\psi_Z})$.
Start with $F(\psi^{AE}_Z,\pi^A\otimes \psi^E)$ and observe that $W\ket{\psi}=VU_X\ket{\psi_Z}$ is a purification of $\pi^A\otimes\psi^E$, as is $U_X\ket{\psi_Z}$. 
Of course, $\ket{\psi_Z}$ is a purification of $\psi^{AE}_Z$, and so  
\begin{align}
\label{eq:qze}
\q(Z^A|E)_{\psi} =\max_V \bra{\psi_Z}(U_X^{A''A|A})^\dagger V^{A'A''B|A'B}\ket{\psi_Z},
\end{align}
where $V$ is an isometry from $A'B$ to $A''A'B$. 
Calling the optimizer $V_X$ and applying it to \eqref{eq:fidelitytriangle} gives \eqref{eq:aqecsecretkey}. 

For \eqref{eq:aqectrump} we use $\q(X^A|A'E)_{\psi_Z}$ to construct an appropriate $V_Z$ in $F(\ket{\psi_Z},V_Z\ket{\psi})$. 
For $\q(X^A|A'E)_{\psi_Z}$ the relevant state is the $X$-measured version of $\psi_Z$:
\begin{subequations}
\begin{align}
\psi_{Z,X}^{AA'BE}
&:=\tr_{A''}[U_X^{AA''|A}\psi_Z^{AA'BE}(U_X^{AA''|A})^\dagger]\\
&\phantom{:}=\tfrac1d\sum_x \ketbra{\tilde x}^A\otimes (Z^{-x})^{A'}\psi^{A'BE} (Z^{x})^{A'}.\label{eq:twicemeasured}
\end{align}
\end{subequations}
Tracing out $A$ dephases $A'$ in the $Z$ basis, meaning $\psi_{Z,X}^{A'BE}=\psi_Z^{A'BE}$. 
Note that a controlled-phase operation from $A$ to $A'$ transforms $\psi_{Z,X}^{AA'BE}$ into $\mix^A\otimes \psi^{A'BE}$. 
Hence, 
\begin{subequations}
\begin{align}
 \q(X^A|A'E)_{\psi_Z}
 &=F(\psi^{A'E},\psi_Z^{A'E})\\
 &=\max_V \bra{\psi_Z}V^{A'B|B}\ket{\psi}^{ABE},
 \end{align}
 \end{subequations}
 where in the final step we have interchanged the $A$ and $A'$ labels. 
 Calling the optimizer $V_Z$ and applying it to \eqref{eq:fidelitytriangle} with $V_X$ from the previous argument gives \eqref{eq:aqectrump}. 
\end{proof}


\section{Version 2: Commuting guesses}
\noindent The uncertainty principle implies that Bob cannot always win version 2 of the uncertainty game, for to do so would require preparing a state that is simultaneously an eigenstate of $X$ and $Z$. 
In fact, we can obtain a quantitative approximate statement in this direction from the above results. 
Using Bob's optimal measurement $V_X$ in \eqref{eq:qze} we obtain $\q(Z^A|E)_{\psi}\geq P(X^A|BA')_{\psi_Z}$, and by Lemma~\ref{lem:connectpguess} this implies $\q(Z^A|E)_{\psi}\geq P(X^A|B)_\psi$. 
Thus, the larger Bob's probability of guessing $X$ using $B$, the more the distribution of $Z$ looks uniform and independent of the system $E$. 
A tighter relation comes from an entropic uncertainty relation for min- and max-entropies~\cite{tomamichel_uncertainty_2011}, which in the present notation reads  
\begin{align}
\label{eq:guessingfidelity}
\max_\sigma F(\psi^{AE}_Z,\mix^A\otimes \sigma^E)^2\geq \pg(X^A|B)_\psi.
\end{align}

\begin{figure}[th!]
\centering
\includegraphics{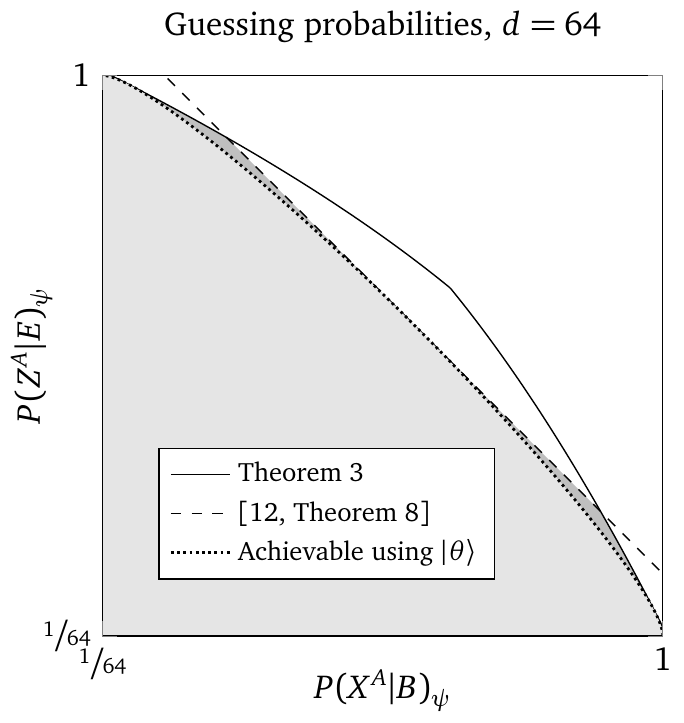}
\caption{\label{fig:tripartite} Feasible and achievable guessing probabilities for version 2 of the uncertainty game, for the case $d=64$.}
\end{figure}

But what are the possible combinations of guessing probabilities? 
The fidelity quantity $Q$ does have an operational meaning, but is not so immediately related to the guessing game. 
It turns out that \eqref{eq:guessingfidelity} can be transformed into a constraint on the set of possible guessing probabilities. 
In particular, we have the following bounds.

\begin{theorem}
\label{prop:tripartite}
For any tripartite state $\psi^{ABE}$, 
\begin{align}
\pg(Z^A|E)_\psi+\left( \pg(X^A|B)_\psi-\tfrac 1d\right)^2&\leq 1,\label{eq:tripartite1}\\
\pg(X^A|B)_\psi+\left( \pg(Z^A|E)_\psi-\tfrac 1d\right)^2&\leq 1.\label{eq:tripartite2}
\end{align}
\end{theorem}

The proof is based on a bound relating the trace distance and fidelity.
Suppose $\{\Lambda_z^E\}$ is the optimal measurement in $\pg(Z^A|E)_\psi$ and define $\Upsilon^{AE}=\sum_z \ketbra z^A\otimes \Lambda_z^E$.
Then $\tr[\Upsilon^{AE}(\psi^{AE}_Z-\mix^A\otimes \sigma^E)  ]=\pg(Z^A|E)_\psi-\tfrac 1d$. 
Maximizing this expression over all possible POVM elements gives the trace distance: $\delta(\rho,\sigma):=\max_\Upsilon\tr[\Upsilon(\rho-\sigma)]$, and so we can appeal to the bound $\delta(\rho,\sigma)^2+F(\rho,\sigma)^2\leq 1$~\cite[\S9.2.3]{nielsen_quantum_2000} to infer 
\begin{align}
\label{eq:Hminlowerbound}
\left(\pg(Z^A|E)_\psi-\tfrac 1d\right)^2+F(\psi_Z^{AE},\pi^A\otimes \sigma^E)^2\leq 1.
\end{align}
Using \eqref{eq:guessingfidelity} in \eqref{eq:Hminlowerbound}, choosing $\sigma^E$ to be the fidelity optimizer, immediately gives \eqref{eq:tripartite2}. 
The other inequality follows by interchanging observables and $B$ and $E$ systems.

Theorem~\ref{prop:tripartite} tells us more precisely how the probability of guessing the outcome of one observable tends to its minimum as the probability of guessing the other goes to unity, as illustrated in Fig.~\ref{fig:tripartite} for $d=64$.
The bounds are nearly tight when one guessing probability is large: Bob can simply interpolate between the $X$ and $Z$ bases by preparing a state in the family $\ket{\theta}=\tfrac1{\sqrt{\mathcal N}}\left(\cos\theta\ket{0}+\sin\theta\ket{\tilde 0}\right)$, for $\theta\in [0,\tfrac\pi2]$ and $\mathcal N$ the appropriate normalization, and always guess the outcomes will be $Z=0$ and $X=0$. 

The bounds are loose for $P(Z^E|B)\approx P(X^A|B)$, but here we can appeal to bounds for a related uncertainty game. 
Instead of one party guessing Alice's outcome, \cite{tomamichel_monogamyofentanglement_2013} supposes there are two, and each is told which observable was measured. 
This is more information than Bob receives in version 2 of the present game, so the guessing probabilities here must be smaller. 
Nevertheless, by having two parties, there exist commuting guessing measurements for the two observables, meaning constraints derived in \cite{tomamichel_monogamyofentanglement_2013} also apply here. 
Figure~\ref{fig:tripartite} shows that their Theorem 8 together with Theorem~\ref{prop:tripartite} give a nearly-tight characterization of allowed guessing probabilities. 


For qubits, the situation is even better, as we can appeal to a bound from~\cite{coles_unifying_2015} which links the min-max uncertainty relation \eqref{eq:guessingfidelity} to wave-particle duality relations. 
For $d=2$ using their Eq.~(6) instead of \eqref{eq:Hminlowerbound} in \eqref{eq:guessingfidelity} gives $(2\pg(Z^A|E)_\psi-1)^2+(2\pg(X^A|B)_\psi-1)^2\leq 1$, which precisely matches the achievable strategy given above. 
Moreover, since \eqref{eq:Hminlowerbound} is tighter than Eq.~(6) for $d>2$, Theorem~\ref{prop:tripartite} leads to a tightened version of the wave-particle duality relation for symmetric interferometers in \cite[Theorem 1]{coles_unifying_2015} by using the definitions therein of the particle distinguishability $\mathcal D=(d\pg(Z^A|E)_\psi-1)/(d-1)$ and the visibility $\mathcal V=\max_{X}(d\pg(X^A|B)_\psi-1)/(d-1)$, where the maximization is over all  observables conjugate to $Z$.

\section{Conclusions and open questions}
\noindent I have given uncertainty relations in the form of bounds on the guessing probabilities in the two variants of the uncertainty game. 
The uncertainty relation for the first version yields a sufficient condition for approximate quantum error correction, and a simple modification of the proof yields two stronger but nonconstructive sufficient conditions. 
In combination with \cite[Theorem 8]{tomamichel_monogamyofentanglement_2013}, the bounds on the second version were found to be essentially tight, but tightness of the first is an open question. 
Furthermore, following the approach of \cite{coles_unifying_2015}, the relation for the second version yields a new wave-particle duality relation for multi-path interferometers. 

Finally, it is interesting to consider if channel versions of the uncertainty relations could hold. 
For instance, we may ask if reliable transmission of classical $X$ and $Z$ information even in the worst-case implies that the given channel is close to the identity channel. 
However, a counterexample constructed in~\cite{renes_heisenberg_2016} shows that this is not the case. 
The particular channel is such that Bob's probabilities of guessing $X$ and $Z$ in the worst case are $(d+\sqrt 2-2)^2/d^2$ and exactly 1, respectively, and yet two particular channel inputs lead to completely distinct outputs to the channel environment. 
However, the information-disturbance tradeoff of \cite{kretschmann_information-disturbance_2008} requires the environment output of channels close to identity to be essentially independent of the input. 
Hence, no channel analog of Theorem~\ref{bipartiteresult} or of \eqref{eq:guessingfidelity} can hold, though neither statement of Theorem~\ref{thm:altbipartite} is apparently ruled out.
Moreover, the worst-case probability of guessing $Z$ from the output to the environment is $1/(d-1)$, and therefore \eqref{eq:tripartite2} (the tighter of the two bounds in this case) does hold. 
It would be interesting to determine if worst-case versions of Theorems~\ref{thm:altbipartite} and \ref{prop:tripartite} hold for channels generally; doubly so for the latter since it is derived from \eqref{eq:guessingfidelity}, which we just observed does not hold in this setting. 

\vspace{2mm}
{\bf Acknowledgements:} 
I am grateful to Patrick Coles for pointing out the bound from~\cite{coles_unifying_2015}. 
This work was supported by the Swiss National Science Foundation (through the NCCR `Quantum Science and Technology' and grant No. 200020-135048) and the European Research Council (grant No. 258932).

\printbibliography[heading=bibintoc,title=References]

\end{document}